\documentclass[11pt,notref,notcite,a4paper]{article}
\usepackage{amsfonts}
\usepackage{graphics,amsmath,amssymb,amsthm}
\usepackage{stmaryrd}
\def \hh{\vskip0.5\baselineskip \hbox to \hsize}

\newtheorem{Lemma}{Lemma}
\newtheorem{Theorem}{Theorem}

\numberwithin{equation}{section}
\begin{document}

%===================================================TITLE=================================================
\title{Symmetries and Lie algebra of the differential-difference
Kadomstev-Petviashvili hierarchy}

\vskip 10pt

\author{ Xian-long Sun\footnote{E-mail address: xlongs@yahoo.cn},~~Da-jun
Zhang\footnote{Corresponding author. E-mail address:
djzhang@staff.shu.edu.cn},~~Xiao-ying Zhu,~~Deng-yuan Chen
\\
{\small \it Department of Mathematics, Shanghai University,
Shanghai 200444, P.R.China} }
%\date{}
\maketitle

%===================================================Abstract==============================================
\begin{abstract}

By introducing suitable non-isospectral flows we construct two sets
of symmetries for the isospectral differential-difference
Kadomstev-Petviashvili hierarchy. The symmetries form an infinite
dimensional Lie algebra.

%==================================================Keywords================================================
\end{abstract}
%==============================================INTRODUCTION=================================================

\section{Introduction}

Searching for symmetries and Lie algebraic structure is an important
and interesting topic in integrable systems\cite{Fokas-87}. Variety
of  methods have been developed  to obtain infinitely many
symmetries and their Lie algebraic structures for Lax integrable
systems\cite{Fuch-83}-\cite{CHEN-CHAOS-2003}, for both
(1+1)-dimensional and high-dimensional cases.
%Ma-JMP-1992, Ma-Fuchssteiner-JMP-1999,
%Ma-Tamizhimani-JPSJ-1999, Tian-1990}.
One of efficient ways is to use Lax representation of isospectral
and non-isospectral flows (cf.
\cite{Ma-90}-\cite{Ma-Tamizhimani-JPSJ-1999},\cite{Chen-91,CHEN-JMP-1996}),
rather than recursion operators, and this approach has been extended
to high dimensional continuous integrable
systems\cite{CHEN-CHAOS-2003}.

In this paper, we consider the symmetries and their Lie algebra for
the differential-difference Kadomstev-Petviashvili (D$\Delta$KP)
hierarchy by means of Lax representation approach.
The (isospectral)
D$\Delta$KP hierarchy is derived following the basic frame of Sato's
theory starting from a quasi-difference operator
\cite{S-1997,S-DOC-1998}. However, due to the lack of a neat form of
discrete derivatives, the non-isospectral flows do not have regular
asymptotic properties as the continuous non-isospectral flows do. We
have to choose suitable time evolution for spectral parameter so
that we can get suitable non-isospectral flows which can be used to
construct symmetries.

%theory we derive the
%isospectral and non-isospectral hierarchies,
%then by the implicit representations of the isospectral flow
%$\{K_m\}$ and non-isospectral flow $\{\sigma_n\}$, we construct the
%symmetries and their algebraic structure for isospectral and
%non-isospectral differential-difference Kadomstev-Petviashvili
%(D$\Delta$KP) hierarchies.

This paper is organized as follows. In Sec.2, we construct isospectral and -non-isospectral
 D$\Delta$KP hierarchies from a quasi-difference operator. In Sec.3,
we construct two sets of symmetries and their Lie algebra for the isospectral D$\Delta$KP hierarchy.
% introduce the implicit representations for the flows
%and shows that the flows and related symmetries both constitute
%infinite dimensional Lie algebra.

%================================= Sec. 2 =========================================

\section{The isospectral and non-isospectral D$\Delta$KP hierarchies}

~~~~~Let us consider the difference analogy of a quasi-differential
operator\cite{S-1997}
\begin{equation}
L=\Delta+u_{0}+u_{1}\Delta^{-1}+\cdots+u_{j}\Delta^{-j}+\cdots,\label{L}
\end{equation}
where $u_s=:u_s(n,t)=u_s(n,t_{1},t_{2},\cdots)\
(s=0,1,2,\ldots)$, $t=(t_1,t_2,\cdots)$, each $u_s$ varnishes rapidly when $|n|\to\infty$, $\Delta$ denotes the forward difference operator
defined by $\Delta f(n)=(E-1)f(n)=f(n+1)-f(n)$ and the shift
operator $E$ defined by $Ef(n)=f(n+1)$. The operators $\Delta$ and
$E$ are connected by $\Delta=E-1$ and
$\Delta\Delta^{-1}=\Delta^{-1}\Delta=1$. Obviously, the
$r^{th}$-power of $L$ can be expressed as
\begin{subequations}
\begin{equation}
L^{r}=\sum\limits_{j\leq r}p_{r,j}(u)\Delta^{j}.\label{a-1}
\end{equation}
where the coefficients $p_{r,j}(u)$ are uniquely determined by the
coordinates $u_{j}\ (j=0,1,2,\ldots)$ and their differences. Here by
$u$ we denote $(u_0,u_1,\cdots)^T$. $L^{r}$ can be seperated into
\begin{equation}
(L^{r})_{+}=\sum\limits_{j=0}^{n}p_{r,j}(u)\Delta^{j}, \ \
(L^{r})_{-}=L^{r}-(L^{r})_{+},
\end{equation}
\end{subequations}
where $( ~)_{+}$ denotes the nonnegative part of $\Delta$ and $(
~)_{-}$ the residual part.

In general, the isospectral flows can be obtained from the
compatibility of
\begin{subequations}
\begin{eqnarray}
L\phi=\eta\phi,\\
\phi_{t_{s}}=A_{s}\phi;
\end{eqnarray}
\end{subequations}
i.e
\begin{subequations}
\begin{equation}
L_{t_{s}}=[A_{s},L],
\label{Lax-eq}
\end{equation}
where $\eta_{t_{s}}=0$, $A_{s}=(L^{s})_{+}$ and obviously $A_{s}$ satisfies
the boundary condition
\begin{equation}
A_{s}|_{u=0}=\Delta^{s}.
\end{equation}
\end{subequations}
The first few explicit forms of $A_{s}$ and equations given by the
Lax equation \eqref{Lax-eq} are\cite{S-DOC-1998}
\begin{subequations}\label{2.5}
\begin{eqnarray}
A_{1}&=&\Delta+u_{0},\label{A1}\\
A_{2}&=&\Delta^{2}+({\Delta}u_{0}+2u_{0})\Delta+({\Delta}u_{0}+u_{0}^{2}+{\Delta}u_{1}+2u_{1}),\label{A2}\\
A_{3}&=&\Delta^{3}+a_{1}\Delta^{2}+a_{2}\Delta+a_{3},\label{A3}\\
&&\cdots,\nonumber
\end{eqnarray}
\end{subequations}
with
\begin{subequations}
\begin{eqnarray}
a_{1}&=&\Delta^{2}u_{0}+3{\Delta}u_{0}+3u_{0},\\
a_{2}&=&2\Delta^{2}u_{0}+3{\Delta}u_{0}+3u_{0}^{2}+3u_{0}{\Delta}u_{0}
+({\Delta}u_{0})^{2}+3u_{1}+3{\Delta}u_{1}+\Delta^{2}u_{1},\\
a_{3}&=&\Delta^{2}u_{0}+5u_{0}u_{1}+3u_{0}{\Delta}u_{0}+u_{0}^{3}
+({\Delta}u_{0})^{2}+{\Delta}u_{0}{\Delta}u_{1}+3u_{0}{\Delta}u_{1}\nonumber\\
&&+u_{1}{\Delta}u_{0}+u_{1}E^{-1}u_{0}+2\Delta^{2}u_{1}+3{\Delta}u_{1}+3u_{2}+3{\Delta}u_{2}+\Delta^{2}u_{2};
\end{eqnarray}
\end{subequations}
\begin{subequations}\label{2.7}
\begin{eqnarray}
u_{0,t_{1}}&=&q_{10}={\Delta}u_{1},\label{q10}\\
u_{1,t_{1}}&=&q_{11}={\Delta}u_{1}+{\Delta}u_{2}+u_{0}u_{1}-u_{1}E^{-1}u_{0},\label{q11}\\
u_{2,t_{1}}&=&q_{12}={\Delta}u_{3}+{\Delta}u_{2}+u_{0}u_{2}+u_{1}E^{-1}u_{0}-u_{2}E^{-2}u_{0}-u_{1}E^{-2}u_{0},\label{q12}\\
&&\cdots;\nonumber
\end{eqnarray}
\end{subequations}
\begin{subequations}\label{2.8}
\begin{eqnarray}
u_{0,t_{2}}=q_{20}\label{u0-2}
&=&\Delta^{2}u_{1}+2{\Delta}u_{2}+\Delta^{2}u_{2}+u_{1}{\Delta}u_{0}+2u_{0}{\Delta}u_{1}\nonumber\\
&&+({\Delta}u_{0}){\Delta}u_{1}+u_{0}u_{1}-u_{1}E^{-1}u_{0},\\
u_{1,t_{2}}=q_{21}\label{u1-2}
&=&{\Delta}^{2}u_{1}+2{\Delta}u_{2}+2{\Delta}^{2}u_{2}+2{\Delta}u_{3}+{\Delta}^{2}u_{3}+2u_{0}{\Delta}u_{1}+{\Delta}u_{0}{\Delta}u_{1}\nonumber \\
&&+2u_{0}u_{2}+u_{2}{\Delta}u_{0}+2u_{0}{\Delta}u_{2}+{\Delta}u_{0}{\Delta}u_{2}+u_{1}{\Delta}u_{0}+u_{1}u_{0}^{2}+u_{1}^{2} \nonumber\\
&&+u_{1}{\Delta}u_{1}-u_{1}E^{-2}u_{0}+u_{1}E^{-1}u_{0}-u_{1}E^{-1}u_{1}-u_{2}E^{-1}u_{0}\nonumber\\
&&-u_{2}E^{-2}u_{0}-u_{1}E^{-1}u_{0}^{2},\\
&&\cdots.\nonumber
\end{eqnarray}
\end{subequations}
%\begin{eqnarray}
%u_{0,t_{3}}=q_{30}\label{u0-3}
%&=&{\Delta}^{3}u_{3}+3{\Delta}^{2}u_{3}+3{\Delta}u_{3}\nonumber\label{u0-3}\\
%&&+3{\Delta}^{3}u_{2}+(6+a_{1}){\Delta}^{2}u_{2}+5{\Delta}u_{2}+u_{2}-u_{2}E^{-2}a_{1}\nonumber\\
%&&+3{\Delta}^{3}u_{1}+(3+2a_{1}){\Delta}^{2}u_{1}+(2a_{1}+a_{2}){\Delta}u_{1}+a_{2}u_{1}+u_{1}E^{-1}a_{1}\nonumber\\
%&&-u_{1}E^{-2}a_{1}-u_{1}E^{-1}a_{2}+{\Delta}^{3}u_{0}+a_{1}{\Delta}^{2}u_{0}+a_{2}{\Delta}u_{0}-{\Delta}a_{3},\\
%&&\cdots.\nonumber
%\end{eqnarray}

From \eqref{2.7}, we obtain
\begin{subequations}\label{2.9}
\begin{eqnarray}
u_{1}&=&\Delta^{-1}\frac{\partial{u_{0}}}{\partial{t_{1}}},\label{u-1}\\
u_{2}&=&\Delta^{-2}\frac{\partial^{2}{u_{0}}}{\partial{t_{1}^{2}}}-\Delta^{-1}\frac{\partial{u_{0}}}{\partial{t_{1}}}-E^{-1}u_{0}\Delta^{-1}\frac{\partial{u_{0}}}
{\partial{t_{1}}}+\Delta^{-1}(u_{0}\frac{\partial{u_{0}}}{\partial{t_{1}}}),\label{u-2}\\
%u_{3}&=&\Delta^{-1}\{\Delta^{-2}\frac{\partial^{3}{u_{0}}}{\partial{t_{1}^{3}}}-2\Delta^{-1}\frac{\partial^{2}{u_{0}}}
%{\partial{t_{1}^{2}}}-E^{-1}\frac{\partial{u_{0}}}{\partial{t_{1}}}\Delta^{-1}\frac{\partial{u_{0}}}{\partial{t_{1}}}
%-E^{-1}u_{0}\Delta^{-1}\frac{\partial^{2}{u_{0}}}{\partial{t_{1}^{2}}}\nonumber\\
%&&+\Delta^{-1}(u_{0}\frac{\partial^{2}{u_{0}}}{\partial{t_{1}^{2}}})+\Delta^{-1}(\frac{\partial{u_{0}}}
%{\partial{t_{1}}})^{2}+\frac{\partial{u_{0}}}{\partial{t_{1}}}+2u_{0}\Delta^{-1}\frac{\partial{u_{0}}}
%{\partial{t_{1}}}-u_{0}\Delta^{-2}\frac{\partial^{2}{u_{0}}}{\partial{t_{1}^{2}}}\nonumber\\
%&&+u_{0}E^{-1}u_{0}\Delta^{-1}\frac{\partial{u_{0}}}{\partial{t_{1}}}-u_{0}\Delta^{-1}(u_{0}\frac{\partial^{2}
%{u_{0}}}{\partial{t_{1}^{2}}})-2\Delta^{-1}\frac{\partial{u_{0}}}{\partial{t_{1}}}E^{-1}u_{0}
%+[\Delta^{-2}\frac{\partial^{3}{u_{0}}}{\partial{t_{1}^{3}}}\nonumber\\
%&&-2\Delta^{-1}\frac{\partial^{2}{u_{0}}}{\partial{t_{1}^{2}}}+2\Delta^{-1}\frac{\partial{u_{0}}}
%{\partial{t_{1}}}-E^{-1}\frac{\partial{u_{0}}}{\partial{t_{1}}}\Delta^{-1}\frac{\partial{u_{0}}}
%{\partial{t_{1}}}-E^{-1}u_{0}\Delta^{-1}\frac{\partial^{2}{u_{0}}}{\partial{t_{1}^{2}}}\nonumber\\
%&&+\Delta^{-1}(u_{0}\frac{\partial^{2}{u_{0}}}{\partial{t_{1}^{2}}})](E^{-2}u_{0})\},\label{u-3}\\
&&\cdots.\nonumber
\end{eqnarray}
\end{subequations}
Eliminating $u_{1},u_{2}, \cdots$ from \eqref{q10}, \eqref{u0-2},
$\cdots$, one can obtain ($u_{0}=u,t_{1}=y$)\cite{S-1997,S-DOC-1998}
\begin{subequations}\label{2.10}
\begin{eqnarray}
u_{t_{1}}=K_{1}&=&u_{y},\label{k1}\\
u_{t_{2}}=K_{2}&=&(1+2\Delta^{-1})u_{yy}-2u_{y}+2uu_{y},\label{ddkp}\\
%u_{t_{3}}=K_{3}&=&{\Delta}^{3}\bar{u}_{3}+3{\Delta}^{2}\bar{u}_{3}+3{\Delta}\bar{u}_{3}\nonumber\\
%&&+3{\Delta}^{3}\bar{u}_{2}+(6+\tilde{a}_{1}){\Delta}^{2}\bar{u}_{2}+5{\Delta}\bar{u}_{2}+\bar{u}_{2}-
%\bar{u}_{2}E^{-2}\tilde{a}_{1}\nonumber\\
%&&+3{\Delta}^{3}\bar{u}_{1}+(3+2\tilde{a}_{1}){\Delta}^{2}\bar{u}_{1}+(2\tilde{a}_{1}+\tilde{a}_{2})
%{\Delta}\bar{u}_{1}+\tilde{a}_{2}\bar{u}_{1}+\bar{u}_{1}E^{-1}\tilde{a}_{1}\nonumber\\
%&&-\bar{u}_{1}E^{-2}\tilde{a}_{1}-\bar{u}_{1}E^{-1}\tilde{a}_{2}+{\Delta}^{3}u+\tilde{a}_{1}{\Delta}^{2}u
%+\tilde{a}_{2}{\Delta}u-{\Delta}\tilde{a}_{3}\label{k3}\\
&&\cdots,\nonumber
\end{eqnarray}
\end{subequations}
%with\\
%
%\begin{equation*}
%\tilde{a}_{1}=a_{1}|_{u_{0}=u,t_{1}=y},~~~~~~~
%    \tilde{a}_{2}=a_{2}|_{u_{0}=u,t_{1}=y},~~~~~~~
%\tilde{a}_{3}=a_{3}|_{u_{0}=u,t_{1}=y},
%\end{equation*}
%\begin{equation*}
%\bar{u}_{1}=u_{1}|_{u_{0}=u,t_{1}=y},~~~~~~~
%  \bar{u}_{2}=u_{2}|_{u_{0}=u,t_{1}=y},~~~~~~~
%  \bar{u}_{3}=u_{3}|_{u_{0}=u,t_{1}=y}.
%\end{equation*}
which are isospectral D$\Delta$KP hierarchy where Eq.\eqref{ddkp} is
the well-known D$\Delta$KP equation.

To get the $\tau$-symmetries we need to introduce the
non-isospectral D$\Delta$KP hierarchy. In this case, we
set\footnote{ One may wonder that \eqref{eta-t} is a linear
combination and so is the non-isospectral flow $\sigma_r$. Actually,
in the Lax representation approach we need $\sigma_{r}|_{u=0}=0$.
Suppose that we start from a general form
\begin{equation*}
\eta_{t_{r}}=a\eta^{\alpha}+b\eta^{\beta},
\label{eta-t-1}
\end{equation*}
with constants $a,b$ and integers $\alpha,\beta$.
Then the Lax equation is
\begin{equation*}
L_{t_{r}}=[B_{r},L]+aL^{\alpha}+bL^{\beta}.\label{non-lax-equa-1}
\end{equation*}
Noting that $L|_{u=0}=\Delta$ and \eqref{Br-0} the r.h.s. of the above equation becomes
$$-\Delta^{r}-\Delta^{r-1}+a\Delta^{\alpha }+b\Delta^{\beta}$$
when $u=0$,
which means we have to take $a=b=1, \alpha=r,\beta=r-1$ so that it vanishes.
Hence we need the time evolution \eqref{eta-t}.}
\begin{equation}
\eta_{t_{r}}=\eta^{r}+\eta^{r-1}.
\label{eta-t}
\end{equation}
Then the Lax equation turns out to be
\begin{subequations}
\begin{equation}
L_{t_{r}}=[B_{r},L]+L^{r}+L^{r-1},\label{non-lax-equa}
\end{equation}
where
\begin{equation}
B_{r}=b_{0}\Delta^{r}+b_{1}\Delta^{r-1}+\cdots+b_{r}, \
(r>0)\label{Br}
\end{equation}
in which $b_{i}~ (i=0,1,2,\cdots,r)$ are undetermined functions of
coordinates $u_{j}~(j=0,1,2,\ldots)$ and their differences. $B_{r}$
is imposed the boundary condition
\begin{equation}
B_{r}|_{u=0}=t_{1}\Delta^{r}+n\Delta^{r-1}, \label{Br-0}
\end{equation}
\end{subequations}
and then the both sides of the Lax equation \eqref{non-lax-equa} go
to zero when $u\to 0$.

The first few of $B_{r}$ and equations given by \eqref{non-lax-equa}
are
\begin{subequations}
\begin{eqnarray}
B_{1}&=&t_{1}A_{1}+n,\\
B_{2}&=&t_{1}A_{2}+n\Delta+nu_{0}+\Delta^{-1}u_{0},\\
B_{3}&=&t_{1}A_{3}+n\Delta^{2}+(2nu_{0}+n{\Delta}u_{0}+\Delta^{-1}u_{0})\Delta+u_{0}\Delta^{-1}u_{0}+2nu_{1}
+n{\Delta}u_{0}\nonumber\\
&&+n{\Delta}u_{1}-2u_{1}-{\Delta}u_{1}+nu_{0}^{2}-u_{0}^{2}+\Delta^{-1}(u_{1}-u_{0}+u_{0}^{2}),\\
&&\cdots;\nonumber
\end{eqnarray}
\end{subequations}
\begin{subequations}
\begin{eqnarray}
u_{0,t_{1}}&=&t_{1}q_{10}+u_{0},\label{non-u01}\\
u_{1,t_{1}}&=&t_{1}q_{11}+2u_{1},\\
u_{2,t_{1}}&=&t_{1}q_{12}+u_{1}+3u_{2},\\
&&\cdots;\nonumber
\end{eqnarray}
\end{subequations}
\begin{subequations}
\begin{eqnarray}
u_{0,t_{2}}&=&t_{1}q_{20}+n{\Delta}u_{1}+u_{0}^{2}-u_{0}+3u_{1}+{\Delta}u_{1},\label{non-u02}\\
u_{1,t_{2}}&=&t_{1}q_{21}+n{\Delta}u_{1}+n{\Delta}u_{2}+(n+1)u_{0}u_{1}
+u_{1}{\Delta}^{-1}u_{0}+2u_{1}+{\Delta}u_{2}\nonumber\\
&&+3u_{2}+(2-n)u_{1}E^{-1}u_{0}-u_{1}E^{-1}\Delta^{-1}u_{0}+{\Delta}u_{1},\\
&&\cdots;\nonumber
\end{eqnarray}\end{subequations}
\begin{eqnarray}
u_{0,t_{3}}&=&t_{1}q_{30}-{\Delta}^{2}u_{0}+nu_{0}{\Delta}u_{1}+n\Delta(u_{0}u_{1})+\Delta^{-1}u_{0}{\Delta}u_{1}
+nu_{0}u_{1}+n\Delta^{2}u_{1}\nonumber\\
&&+u_{1}\Delta^{-1}u_{0}-u_{0}{\Delta}u_{0}-4n{\Delta}u_{1}-2u_{1}-{\Delta}u_{0}-{\Delta}u_{1}+u_{0}-3u_{0}^{2}\nonumber\\
&&-nu_{1}E^{-1}u_{0}+u_{1}E^{-1}u_{0}+u_{0}u_{1}-u_{1}E^{-1}\Delta^{-1}u_{0}+n\Delta^{2}u_{2}\nonumber\\
&&+2{\Delta}u_{2}+2u_{2},\label{non-u03}\\
&&\cdots.\nonumber
\end{eqnarray}

Here $A_{l}$ and $q_{ij}$ are described by \eqref{2.5}, \eqref{2.7}
and \eqref{2.8} respectively.

 Then substituting \eqref{2.9} with $t_{1}=y$ into
 \eqref{non-u01}, \eqref{non-u02} and \eqref{non-u03}  yields($u_{0}=u$)
\begin{subequations}
\begin{eqnarray}
u_{t_{1}}=\sigma_{1}&=&yK_{1}+u,\\
u_{t_{2}}=\sigma_{2}&=&yK_{2}+(1+n)u_{y}+3\Delta^{-1}u_{y}+u^{2}-u,\\
%u_{t_{3}}=\sigma_{3}
%&=&yK_{3}-\Delta^{2}u+2nuu_{y}+u_{y}\Delta^{-1}u+2nu\Delta^{-1}u_{y}\nonumber\\
%&&+\Delta^{-1}u_{y}\Delta^{-1}u-u{\Delta}u-4nu_{y}-4\Delta^{-1}u_{y}-{\Delta}u-3u_{y}+u-3u^{2}\nonumber\\
%&&-2n\Delta^{-1}u_{y}E^{-1}u+\Delta^{-1}u_{y}E^{-1}u-\Delta^{-1}u_{y}E^{-1}\Delta^{-1}u\nonumber\\
%&&+nu_{yy}+2\Delta^{-1}u_{yy}-u\Delta^{-1}u_{y}+2\Delta^{-2}u_{yy}+2\Delta^{-1}(uu_{y}),\\
&&\cdots,\nonumber
\end{eqnarray}
\end{subequations}
in which, $K_{s}$ are given by \eqref{2.10}. These equations
constitute the non-isospectral hierarchy of the D$\Delta$KP system.

The obtained isospectral and non-isospectral D$\Delta$KP hierarchies
can be expressed through Lax equations in the following form
\begin{subequations}\label{2.19}
\begin{eqnarray}
L'[K_{s}]&=&[A_{s},L],\label{Lax-exp11}\\
 A_{s}|_{u=0}&=&\Delta^{s};\label{Lax-exp12}
\end{eqnarray}
\end{subequations}
\begin{subequations}\label{2.20}
\begin{eqnarray}
L'[\sigma_{r}]&=&[B_{r},L]+L^{r}+L^{r-1},\label{Lax-exp21}\\
B_{r}|_{u=0}&=&t_{1}\Delta^{r}+n\Delta^{r-1},\label{Lax-exp22}
\end{eqnarray}
\end{subequations}
which we call Lax representations of flows.

%================================= Sec. 3 =========================================
\section{Lie algebra structure of the D$\Delta$KP system}
~~~~In this section, we begin with a discussion of Gateaux
derivative concerning the quasi-difference operator. Let
$\partial_{t_{j}}=\frac{\partial}{{\partial}t_{j}}$ and $\mathcal
{F}$ denote a linear space constructed by all real functions
$f=f(u)$ depending on $n,t$ and derivatives and differences of $u$.
%$f=f(n,t_{1},t_{2},\ldots,u,\partial_{t}^{\alpha}\Delta^{\beta}u,cdots)$,
%where $\alpha,\beta\geq0\in{\mathbb{Z}}$, $\partial_{t}^{\alpha}=$ and $u=u(n,t)$ are
%all real functions defined over $\mathbb{R}\times\mathbb{Z}$ and
%vanish rapidly.
 $f(u)$ is  $C^{\infty}$ differentiable w.r.t. $t$
and $n$, and vanishes rapidly when $|n|\to\infty$.
The Gateaux derivative of $f(u)\in\mathcal{F}$ in direction
$h\in\mathcal{F}$ w.r.t. $u$ is defined as
\begin{equation}
f'[h]=\frac{d}{d\varepsilon}f(u+{\varepsilon}h)|_{\varepsilon=0},
\end{equation}
from which, $\mathcal{F}$ forms a Lie algebra according to the
following Gateaux commutator
\begin{equation}
{\llbracket}f,g{\rrbracket}=f'[g]-g'[f],
\end{equation}
where $f,g\in\mathcal{F}$. For a quasi-difference operator
\begin{equation}
P(u)=\sum\limits_{j\leq{s}}p_{j}(u)\Delta^{j},
\end{equation}
its Gateaux derivative in direction $h$ with
respect to $u$ is defined by
\begin{equation}
P'[h]=\sum\limits_{j\leq{s}}p_{j}'[h]\Delta^{j}.
\end{equation}

Besides, using
\begin{equation}
(p_{j}')'[f]g=(p_{j}')'[g]f,
\end{equation}
one can get\cite{Olver-book}
\begin{equation}
(P'[f])'[g]-(P'[g])'[f]=P'[{\llbracket}f,g{\rrbracket}].\label{com-relation}
\end{equation}
In addition, it is easy to prove the following lemma.
\begin{Lemma}
For the quasi-difference operator $L$ defined in \eqref{L}, $B$ in
the form \eqref{Br} and $X\in\mathcal {F}$, the equation
\begin{equation}
L'[X]=[B,L], ~~B|_{u=0}=0\label{0-equa}
\end{equation}
only admits zero solution $X=0,B=0$.
\end{Lemma}

Then, from the Lax representations  \eqref{2.19}-\eqref{2.20}, we
have the following property.
\begin{Theorem}
Suppose that
\begin{subequations}
\begin{eqnarray}
{\langle}A_{s},A_{r}{\rangle}&=&A_{s}'[K_{r}]-A_{r}'[K_{s}]+[A_{s},A_{r}],\\
{\langle}A_{s},B_{r}{\rangle}&=&A_{s}'[\sigma_{r}]-B_{r}'[K_{s}]+[A_{s},B_{r}],\\
{\langle}B_{s},B_{r}{\rangle}&=&B_{s}'[\sigma_{r}]-B_{r}'[\sigma_{s}]+[B_{s},B_{r}],\label{3.8c}
\end{eqnarray}
\end{subequations}
then we have
\begin{subequations}
\begin{eqnarray}
L'[\llbracket{K_{s},K_{r}}\rrbracket]&=&[{\langle}A_{s},A_{r}\rangle,L],\label{thr1-a}\\
L'[\llbracket{K_{s},\sigma_{r}}\rrbracket]&=&[{\langle}A_{s},B_{r}\rangle,L],\label{thr1-b}\\
L'[\llbracket{\sigma_{s},\sigma_{r}}\rrbracket]&=&[{\langle}B_{s},B_{r}\rangle,L]+(s-r)L^{s+r-1}+2(s-r)L^{s+r-2}+(s-r)L^{s+r-3},\label{thr1-c}
\end{eqnarray}
\end{subequations}
and
\begin{subequations}
\begin{eqnarray}
{\langle}A_{s},A_{r}\rangle|_{u=0}&=&0,\label{thr1-bon-a}\\
{\langle}A_{s},B_{r}\rangle|_{u=0}&=&s\Delta^{s+r-1}+s\Delta^{s+r-2},\label{thr1-bon-b}\\
{\langle}B_{s},B_{r}\rangle|_{u=0}&=&(s-r)\bigg[t_{1}\Delta^{s+r-1}+(t_{1}+n)\Delta^{s+r-2}+n\Delta^{s+r-3}\bigg].\label{thr1-bon-c}
\end{eqnarray}
\end{subequations}
\end{Theorem}

\begin{proof}
We only prove the equalities \eqref{thr1-c} and
\eqref{thr1-bon-c}, the others can be obtained in a similar way.

Taking the Gateaux derivative of \eqref{Lax-exp21} in the direction
$\sigma_{s}$ with respect to $u$, and noting
\begin{equation}
L^{r'}[\sigma_{s}]=[B_{s},L^{r}]+rL^{s+r-1}+rL^{s+r-2}
\end{equation}
and
\begin{equation}
[[B_{s},B_{r}],L]=[B_{s},[B_{r},L]]-[B_{r},[B_{s},L]],
\end{equation}
we have
\begin{equation}
\begin{split}
(L'[\sigma_{r}])'[\sigma_{s}]=&[B'_{r}[\sigma_{s}],L]+[B_{r},[B_{s},L]]+[B_{r},L^{s}]+[B_{r},L^{s-1}]+[B_{s},L^{r}]+rL^{s+r-1}\\
&+rL^{s+r-2}+[B_{s},L^{r-1}]+(r-1)L^{s+r-2}+(r-1)L^{s+r-3}.\label{prov-1}
\end{split}
\end{equation}
Similarly,
\begin{equation}
\begin{split}
(L'[\sigma_{s}])'[\sigma_{r}]=&[B'_{s}[\sigma_{r}],L]+[B_{s},[B_{r},L]]+[B_{s},L^{r}]+[B_{s},L^{r-1}]+[B_{r},L^{s}]+sL^{s+r-1}\\
&+sL^{s+r-2}+[B_{r},L^{s-1}]+(s-1)L^{s+r-2}+(s-1)L^{s+r-3}.\label{prov-2}
\end{split}
\end{equation}
Then \eqref{prov-1} coupled with \eqref{prov-2} yield
\begin{equation}
(L'[\sigma_{s}])'[\sigma_{r}]-(L'[\sigma_{r}])'[\sigma_{s}]=[{\langle}B_{s},B_{r}\rangle,L]+(s-r)(L^{s+r-1}+2L^{s+r-2}+L^{s+r-3}),\label{prov-3}
\end{equation}
which gives \eqref{thr1-c} by using \eqref{com-relation}. Next, noting that
$K_{s},\sigma_{r}\in\mathcal {F}$, i.e.,
$
K_{s}|_{u=0}=\sigma_{r}|_{u=0}=0$,
from \eqref{3.8c} we obtain \eqref{thr1-bon-c} immediately.
%\begin{equation}
%\begin{split}
%{\langle}B_{s},B_{r}\rangle|_{u=0}&=[B_{s},B_{r}]|_{u=0}\\
%&=(s-r)\bigg[t_{1}\Delta^{s+r-1}+(t_{1}+n)\Delta^{s+r-2}+n\Delta^{s+r-3}\bigg]\\
%\end{split}
%\end{equation}
We complete the proof.
\end{proof}

With the above theorem in hand, the algebraic relation of flows
$K_{s}$ and $\sigma_{r}$ can be derived.
\begin{Theorem}
\label{T:2}
The flows
$K_{s}$ and $\sigma_{r}$ form a Lie algebra with structure
\begin{subequations}
\begin{eqnarray}
{\llbracket}K_{s},K_{r}{\rrbracket}&=&0,\label{thr2-a}\\
{\llbracket}K_{s},\sigma_{r}{\rrbracket}&=&sK_{s+r-1}+sK_{s+r-2},\label{thr2-b}\\
{\llbracket}\sigma_{s},\sigma_{r}{\rrbracket}&=&(s-r)(\sigma_{s+r-1}+\sigma_{s+r-2}),\label{thr2-c}
\end{eqnarray}
\end{subequations}
where $s,r\geq 1$ and we set $K_0=\sigma_0=0$.
\end{Theorem}
\begin{proof} In the light of \eqref{0-equa} only admitting zero
solution, \eqref{thr1-a} coupled with \eqref{thr1-bon-a} possesses
the same property as well, which
means \eqref{thr2-a} holds.

Next, taking
\begin{subequations}
\begin{eqnarray}
\theta&=&{\llbracket}K_{s},\sigma_{r}{\rrbracket}-sK_{s+r-1}-sK_{s+r-2},\\
\tilde{A}&=&{\langle}A_{s},B_{r}{\rangle}-sA_{s+r-1}-sA_{s+r-2},
\end{eqnarray}
\end{subequations}
it then follows from \eqref{thr1-b}, \eqref{thr1-bon-b} and the
isospectral Lax representation \eqref{2.19}
that
\begin{equation}
L'[\theta]=[\tilde{A},L],~~\tilde{A}|_{u=0}=0,
\end{equation}
which has only zero solution $\theta=0$ and $\tilde{A}=0$, and then
means \eqref{thr2-b} is true.

Similarly, taking
\begin{subequations}
\begin{eqnarray}
\omega&=&{\llbracket}\sigma_{s},\sigma_{r}{\rrbracket}-(s-r)(\sigma_{s+r-1}+s\sigma_{s+r-2}),\\
\tilde{B}&=&{\langle}B_{s},B_{r}{\rangle}-(s-r)(B_{s+r-1}+B_{s+r-2}),
\end{eqnarray}
\end{subequations}
and noting that $\tilde{B}|_{u=0}=0$ together with \eqref{2.20},
\eqref{thr1-c} and \eqref{thr1-bon-c},
%\begin{eqnarray}
%{\langle}B_{s},B_{r}\rangle|_{u=0}=(s-r)(B_{s+r-1}+B_{s+r-2})|_{u=0}.\label{prov-4}
%\end{eqnarray}
%\eqref{thr1-c}, \eqref{thr1-bon-c} and \eqref{prov-4} together with
%the non-isospectral Lax representation \eqref{Lax-exp21} and
%\eqref{Lax-exp22} give
we then have
\begin{equation}
L'[\omega]=[\tilde{B},L],~~\tilde{B}|_{u=0}=0.
\end{equation}
Hence we get $\omega=0$ and $\tilde{B}=0$, which shows that
\eqref{thr2-c}
is also correct.
Thus we complete the proof.
\end{proof}

Based on Theorem \ref{T:2}, the symmetries and their
algebraic structure for the isospectral D$\Delta$KP hierarchy
$u_{t_{s}}=K_{s}$  can be derived immediately.

\begin{Theorem}
The isospectral D$\Delta$KP hierarchy
$u_{t_{s}}=K_{s}$ can have tow sets of
symmetries,   $K$-symmetries $\{K_{l}\}$ and
$\tau$-symmetries,
$\tau^{s}_{r}=st_{s}K_{s+r-1}+st_{s}K_{s+r-2}+\sigma_{r}$
$(l=1,2,\ldots, ~r=1,2,\ldots)$,
which form a Lie algebra with structure
\begin{subequations}
\label{new-str}
\begin{eqnarray}
{\llbracket}K_{l},K_{r}{\rrbracket}&=&0,\\
{\llbracket}K_{l},\tau^{s}_{r}{\rrbracket}&=&l(K_{l+r-1}+K_{l+r-2}),\\
{\llbracket}\tau^{s}_{l},\tau^{s}_{r}{\rrbracket}&=&(l-r)(\tau^{s}_{l+r-1}+\tau^{s}_{l+r-2}),
\end{eqnarray}
\end{subequations}
where $l,r,s\geq 1$ and we set $K_0=\tau^{s}_{0}=0$.
Especially for the D$\Delta$KP equation \eqref{ddkp}
its symmetries are $K_{l}$ and $\tau_{r}=2tK_{r+1}+2tK_{r}+\sigma_{r}$.
\end{Theorem}

We end up this section by the following two remarks.
First, the new time-dependence \eqref{eta-t} of the spectral parameter $\eta$
leads to the new algebra structure \eqref{new-str},
which is different from the centreless Kac-Moody-Virasoro algebra (cf.\cite{Gungor-06}) of the
D$\Delta$KP equation given in \cite{S-DOC-1998},
and also different from the centreless Kac-Moody-Virasoro algebra of the KP hierarchy obtained in \cite{CHEN-CHAOS-2003}.
Besides, $\{K_1,K_2,\tau^{s}_{1}\}$ compose a subalgebra. This agrees with the symmetry algebra of the
D$\Delta$KP equation obtained in \cite{Zhang-09} where $\tau^{2}_{1}$
provides an invariability for the D$\Delta$KP equation under a combined Galilean-scalar transformation,
and now $\tau^{2}_{1}$ has got its clear context in the Lax representation approach.
The second remark is on the relation between the  D$\Delta$KP equation and the KP equation.
In fact, the D$\Delta$KP equation was originally proposed by Date, et.al.\cite{Date-dis-II}. It was derived from
a bilinear identity (discretized by partially imposing Miwa's transformation on continuous exponential functions) which is related to the KP hierarchy.
However, since in the discrete exponential the discrete variables (eg. $n,m,l$) appear symmetrically and do not represent dispersion relation as in the continuous one,
therefore there are (sometimes complicated) variable combination and transformation involved in
the continuous limit procedure.
There have been many results on the D$\Delta$KP equation, such as bilinear form\cite{Date-dis-II},
Sato's approach\cite{S-1997}-\cite{S-DOC-1998}, Casoratian solutions\cite{Zhang-09,S-JPA}, gauge transformation and
double Casoratian solutions\cite{He-09}, and also symmetries in the present paper.
The relations between these results and those of the KP equation will be investigated in detail elsewhere in terms of continuous limit.

%================================= Sec. 4 =========================================
\section{Conclusion}

In this paper, by introducing suitable time-dependence $\eta_{t_{r}}=\eta^{r}+\eta^{r-1}$
for the spectral parameter $\eta$, we obtained non-isospectral D$\Delta$KP flows $\{\sigma_r\}$
which satisfy $\sigma_r|_{u=0}=0$.
This enables us to construct $K$-symmetries  and
$\tau$-symmetries for the isospectral D$\Delta$KP hierarchy through
the Lax representation approach. The obtained symmetries are proved to
form a Lie algebra.

%================================= Sec. 5 =========================================

\section*{Acknowledgments}
%The authors are grateful to the referees for their invaluable comments.
This project is supported by the National Natural Science Foundation
of China (10671121) and Shanghai Leading Academic Discipline
Project (No.J50101).

\end{document}